\documentclass[envcountsame,runningheads,a4paper]{llncs}

\usepackage[english]{babel}
\usepackage{xspace}
\usepackage{amsmath,amssymb}
\usepackage{graphicx}
\usepackage{stmaryrd}
\usepackage[ruled, vlined]{algorithm2e}
\DontPrintSemicolon
\usepackage{pgf}
\usepackage{tikz}
\usetikzlibrary{arrows,automata,decorations,shapes,snakes}
\usepackage{verbatim}

\newcommand{\N}{\mathbb{N}}
\newcommand{\Z}{\mathbb{Z}}

\newcommand{\xpt}{\textup{EXPTIME}}
\newcommand{\xpth}{\textup{EXPTIME}-hard}
\newcommand{\xps}{\textup{EXPSPACE}}
\newcommand{\psc}{\textup{PSPACE}-complete}
\newcommand{\psh}{\textup{PSPACE}-hard}
\newcommand{\psp}{\textup{PSPACE}}
\newcommand{\ptm}{\textup{P}}

\let\kw\textit

\def\vass{\textup{VASS}}
\def\wkvass{non-blocking \textup{VASS}}
\def\Wkvass{Non-blocking \textup{VASS}}
\newcommand{\pbsys}[3]{\textup{Reach-\textit{semantics}}$_{#1}^{#2}\ifx&#3&{}\else{(#3)}\fi$}
\newcommand{\pbvass}[3]{\textup{Reach-VASS}$_{#1}^{#2}\ifx&#3&{}\else{(#3)}\fi$}
\newcommand{\pbcs}[3]{\textup{Reach-CS}$_{#1}^{#2}\ifx&#3&{}\else{(#3)}\fi$}
\newcommand{\pbwkvass}[3]{\textup{Reach-NBVASS}$_{#1}^{#2}\ifx&#3&{}\else{(#3)}\fi$}
\def\shrtrgd{short-ranged}
\def\shrtrg{short-range}
\def\emptyone{\emptyset_1}
\def\emptytwo{\emptyset_2}
\def\emptyi{\emptyset}

\newcommand{\circpl}{{\raisebox{1pt}{\scalebox{0.9}{\ensuremath{\bigcirc}}}}}
\newcommand{\boxpl}{\raisebox{-0.5pt}{\scalebox{1.15}{\ensuremath{\Box}}}}
\newcommand{\anypl}{\raisebox{-0.5pt}{\scalebox{1.75}{\ensuremath{\diamond}}}}

\title{On The Complexity of Counter Reachability Games
\thanks{This work is supported by the french Agence Nationale de la Recherche, REACHARD (grant ANR-11-BS02-001).
It will appear in the proceedings of the 7th International workshop on Reachability Problems (RP 2013).}
}

\author{Julien Reichert\thanks{reichert@lsv.ens-cachan.fr}}

\institute{LSV, ENS Cachan, France}

\date{\today}

\begin{document}
\maketitle

\begin{abstract}
Counter reachability games are played by two players on a graph with labelled edges.
Each move consists in picking an edge from the current location and adding its label to a counter vector.
The objective is to reach a given counter value in a given location.
We distinguish three semantics for counter reachability games,
according to what happens when a counter value would become negative:
the edge is either disabled, or enabled but the counter value becomes zero, or enabled.
We consider the problem of deciding the winner in counter reachability games
and show that, in most cases, it has the same complexity under all semantics.
Surprisingly, under one semantics, the complexity in dimension one depends on whether the objective value is zero or any other integer.
\end{abstract}

\section{Introduction}

Counter reachability games are played by two~players, a~Reacher and an~Opponent, on a~counter system.
Such a system is represented by a~labelled directed graph $(Q,E)$,
where $Q$ is a finite set of locations and $E \subseteq Q \times \Z^d \times Q$ is a~set of edges.
The integer $d$ is the dimension of the system.
We associate to a counter system a~vector of $d$ counters, which is updated when an edge $(q,v,q')$ is taken by adding $v$ to it.
The locations are partitioned into a~set $Q_1$ of Reacher locations and a~set $Q_2$ of Opponent locations.
A~configuration in a counter system is a pair (location, counter vector).

A~play is an~infinite sequence $(q_0,v_0) (q_1,v_1) \dots \in (Q \times \Z^d)^\omega$,
starting at a~given initial location $q_0$ with the~initial counter vector $v_0$.
At~any stage $i$, the~owner of the~location $q_i$ chooses an~edge $(q_i,v,q_{i+1})$,
then the next configuration is~$(q_{i+1},v_i+v)$.
The~objective is given by a~subset $C$ of $Q \times \Z^d$:
Reacher wins every play that reaches a~configuration in $C$.
Here, we deal with cases where it is equivalent to consider only subsets $C$ that are singletons.

In many works on counter systems, there are only nonnegative counter values,
e.g., in vector addition systems with states (\vass, in short) \cite{KM69},
an edge is disabled whenever it would make a counter become negative.
In energy games \cite{BFLMS08,FJLS11}, the objective is to bound counter values, especially with~$0$ as lower bound.

In order to capture common behaviours around zero, we consider three semantics for counter systems:
\vspace{-2mm}
\begin{itemize}
\item \kw{$\Z$ semantics}: A counter can have any value in~$\Z$.
\item \kw{\vass\ semantics}: An edge is disabled if taking it would make any counter value become negative.
\item \kw{\wkvass\ semantics}: Every time an edge is taken, negative values are replaced by $0$.
\end{itemize}

The decision problem associated to a counter reachability game is to determine whether Reacher has a~winning strategy.
We study decidability and complexity of this problem under the three semantics.
Most of our results assume that the set of edges is restricted to a subset of $Q \times \{-1,0,1\}^d \times Q$;
we call this the \kw{\shrtrg} property and we say that counter systems are \kw{\shrtrgd}.
Any counter system can be transformed into a \shrtrgd\ counter system at the cost of an exponential blowup,
by splitting the edges with labels not in $\{-1,0,1\}^d$.
However, we need to be careful when we deal with reachability issues,
because a run in the \shrtrgd\ counter system visits configurations
that the corresponding run of the first counter system does not visit.

We prove in Section~$3$ that the decision problem is undecidable for reachability games
on counter systems of dimension~two under the $\Z$ semantics,
by an adaptation of the undecidability proof for reachability games on \vass\ of dimension~two in~\cite{BJK10}.

We prove in Section~$4$ that the decision problem is \psc\ for reachability games
on \shrtrgd\ counter systems of dimension~one under the $\Z$ semantics when the objective is~$(q_f,0)$,
and under the \wkvass\ semantics when the objective is~$(q_f,1)$.
The proof is based on mutual reductions from the decision problem for reachability games on \shrtrgd\ counter systems of dimension~one
under the \vass\ semantics when the objective is~$(q_f,0)$, which has been proved \psc\ in \cite{BJK10}.
The case of a reachability games on \shrtrgd\ counter systems of dimension~one under the \wkvass\ semantics
when the objective is $(q_f,0)$ is considered separately. Surprisingly, the decision problem is then in \ptm.

Without the \shrtrg\ property, we have an immediate \xps\ upper bound for counter reachability games
in dimension~one. There are at least two particular cases of counter reachability games
for which the decision problem is \xpth\ in dimension~one:
countdown games~\cite{JLS07} and robot games~\cite{AR13}.
To~the best of our knowledge, it is not known whether counter reachability games in dimension~one are in \xpt.

\section{Definitions}

When we write ``positive'' or ``negative'', we always mean ``strictly positive'' or ``strictly negative''. We
write $-\N$ for the set of nonpositive integers.

A~\kw{counter system} is a directed graph $(Q,E)$, where $Q$ is a finite set of \kw{locations}
and $E \subseteq Q \times \Z^d \times Q$ is a finite set of \kw{edges}, with $d \in \N \setminus \{0\}$.
The vector in~$\Z^d$ is called the \kw{label} of an edge.
A~\kw{configuration} in a counter system is a pair $(q,x)$, where $q \in Q$ and $x \in \Z^d$.
A~\kw{run} of a counter system $(Q,E)$ is an infinite sequence
$r = (q_0,x_0) (q_1,x_1) \dots$ starting from an arbitrary initial configuration $(q_0,x_0) \in Q \times \Z^d$
and such that $(q_i,x_{i+1}-x_i,q_{i+1}) \in E$ for every $i \in \N$.
A~counter system has the \kw{\shrtrg\ property} if the integers in the labels of the edges are always in~$\{-1,0,1\}$.

A~\kw{counter reachability game} is played by two players, a \kw{Reacher} and an \kw{Opponent}, on a counter system $(Q,E)$.
We partition the set of locations into $Q_1 \biguplus Q_2$; Reacher owns $Q_1$, and Opponent owns $Q_2$.
In our figures, we use \circpl\ to represent Reacher locations,
\boxpl\ to represent Opponent locations and \anypl\ when the owner of the location does not matter.

A~\kw{play} is represented by an infinite path of configurations that players form by moving a token on $(Q,E)$ and updating a counter as follows.
At the beginning, the token is at a location $q_0$ and the counter is initialized with $x_0$, hence the initial configuration is $(q_0,x_0)$.
If the token is at $p \in Q_1$, then Reacher chooses an edge $(p,v,q)$, otherwise Opponent chooses.
The token is moved to $q$, the counter is updated to $x+v$, and the configuration $(q,x+v)$ is appended to the play.
There is a special configuration, called the \kw{objective} of the game, such that
Reacher wins every play that visits the objective.

A \kw{play prefix} starting from the configuration $(q_0,x_0)$
is a finite sequence $(q_0,x_0) (q_1,x_1) \dots (q_k,x_k)$ of configurations in the underlying counter system.
A~\kw{strategy} for a player is a function that takes as argument a play prefix
and returns an edge that is available from the end of the play prefix.
Given an configuration $(q_0,x_0)$, two strategies $s_1$ and $s_2$ for the players,
the \kw{outcome} of these strategies from the configuration is the play starting at $(q_0,x_0)$ and obtained
when each player always chooses edges according to his strategy.
A~strategy $s$ is \kw{winning} for a player, from a given configuration,
if he wins the outcome of $s$ with any strategy of the other player from the configuration.
A configuration $(q_0,x_0)$ in the game is \kw{winning} if Reacher has a winning strategy from $(q_0,x_0)$.
The decision problem associated to a counter reachability game
is to determine whether Reacher has a winning strategy from a configuration in input.

A \kw{Vector Addition System with States} (\vass, in short) is a counter system where
the vectors in the configurations are always nonnegative. In order to maintain this property,
an edge in a \vass\ is disabled if a counter would then become negative.
A \kw{\wkvass} is a counter system where every negative counter value is replaced by $0$.

We introduce a notation for the decision problems that we deal with,
and we write \pbsys{d}{1}{x_f} with the following parameters:
a subscript~$d$ for the dimension, an argument~$x_f$ for the counter value in the objective
and a superscript~$1$ to point out, if present, when the system is \shrtrgd.
The counter value in the objective is also optional.
We omit the location in the objective, because only the counter value is relevant here.
For example, let us look at two notations that appear in the next two sections.
\begin{itemize}
\item The problem of deciding the winner on a counter system of dimension~two with an arbitrary objective is denoted by \pbcs{2}{}{}.
\item The problem of deciding the winner on a \shrtrgd\ \wkvass\ of dimension one with~$1$ as objective value is denoted by \pbwkvass{1}{1}{1}.
\end{itemize}

\section{Counter reachability games in dimension two or more}

\subsection{Reduction from \vass\ to general counter systems}

We present a construction that we use in this section
to prove undecidability of counter reachability games in dimension two,
and in the next section to give lower complexity bounds.

In order to show the reduction from \vass\ to general counter systems,
we simulate in the winning condition the deactivation of edges in \vass, which makes the difference to the $\Z$ semantics.
We here denote by $0_d$ the $d$-dimensional vector $(0,\dots,0)$.

\begin{proposition}\label{VASS_crg}
\pbvass{d}{}{0_d} reduces to \pbcs{d}{}{0_d} in polynomial time for any dimension~$d$.
\end{proposition}

\begin{proof}
Let $(Q,E)$ be a \vass\ of dimension $d$, let $(q_0,x_0)$ and $(q_f,x_f)$ be configurations of $(Q,E)$.
We consider the reachability game on $(Q,E)$ where the objective is $(q_f,x_f)$.

The following hypothesis makes most proofs of this work simpler, without loss of generality.
We assume that $q_f$ is a Reacher location.
Else, we could simply create a Reacher location $q_f'$ that has only one outgoing edge
to $q_f$ with label $(0)$ and choose as objective $(q'_f,x_f)$.

We want to build a general counter system on which Reacher has a winning strategy from a particular configuration if,
and only if, he has a winning strategy from $(q_0,x_0)$ in the \vass.
The key property is that each player must be able to win
whenever his adversary makes a counter value become negative.
We can then simulate the \vass\ semantics.

In order to have this property, let $(Q',E')$ be a counter system
with locations $Q' = Q \cup \{\text{test}_e\ |\ e = (p,v,q) \in E, v \not\in \N^d\} \cup \{\text{check},\text{check}_1,\dots,\text{check}_d\}$,
where Reacher owns $Q_1$, the check locations and exactly the locations test$_e$ for which the source of $e$ belongs to Opponent in $(Q,E)$.
The set of edges $E'$ is obtained from $E$, first by splitting every edge $e = (p,v,q)$ such that $v \not\in \N^d$
into two edges $(p,v,\text{test}_e)$ and $(\text{test}_e,0,q)$,
and second by adding moves from every location test$_e$ to the new locations of $Q'$,
as depicted in Figures~\ref{vass_crg_reach} and~\ref{vass_crg_opp}. 

\begin{figure}
\begin{center}
\begin{tikzpicture}[->,>=stealth',shorten >=1pt,auto,node distance=2.5cm,
                    semithick]
  \tikzstyle{state}=[circle,minimum size=1cm,fill=white,draw=black,text=black,font=\small]
  \tikzstyle{oppstate}=[rectangle,minimum size=1cm,fill=white,draw=black,text=black,font=\small]
  \tikzstyle{anystate}=[diamond,minimum size=1cm,fill=white,draw=black,text=black,font=\small]

  \node[state] (A)                    {$p$};
  \node[oppstate]         (B) [right of=A] {test$_e$};
  \node[state]         (C) [right of=B] {check};
  \node[state]         (D) [right of=C] {$\bot$};
  \node[anystate]         (E) [below of=B] {$q$};

  \path (A) edge node {$(-1,-2)$} (B)
        (B) edge node {$(0,0)$} (C)
            edge node {$(0,0)$} (E)
        (C) edge [loop above] node {$(0,-1)$} (C)
            edge [loop below] node {$(-1,0)$} (C)
            edge node {$(0,0)$} (D)
        (D) edge [loop below] node {$(0,0)$} (D);

\end{tikzpicture}
\end{center}
\caption{Gadget to replace an edge $e = (p,(-1,-2),q)$ from a Reacher location
in the reduction from \pbvass{2}{}{(0,0)} to \pbcs{2}{}{(0,0)}.}
\label{vass_crg_reach}
\end{figure}
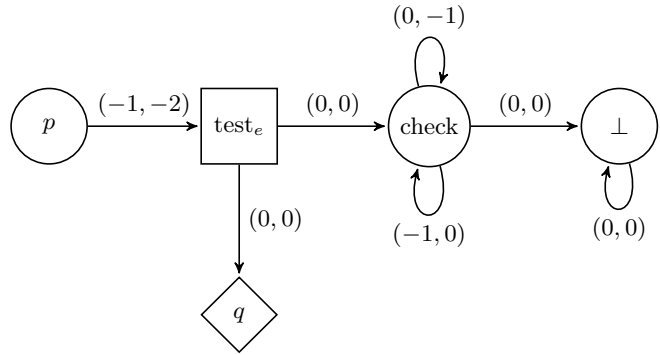

More precisely, $E'$ is the union of the following sets of edges,
where $(x)_{i,d}$ is the vector with $x$ as $i^{\text{th}}$ component and $0$ everywhere else:
\begin{itemize}
\item $\{(p,v,q) \in E\ |\ v \in \N^d\}$;
\item $\{(p,v,\text{test}_e), (\text{test}_e,(0),q)\ |\ e = (p,v,q) \in E, v \not\in \N^d\}$;
\item $\{(\text{test}_e,(0),\text{check})\ |\ e = (p,v,q) \in E, p \in Q_1\}$;
\item $\{(\text{test}_e,(0),\text{check}_i)\ |\ e = (p,v,q) \in E, p \in Q_2, 1 \le i \le d\}$;
\item $\{(\text{check},(-1)_{i,d},\text{check})\ |\ 1 \le i \le d\}$;
\item $\{(\text{check}_i,(-1)_{j,d},\text{check}_i)\ |\ 1 \le j \le d, j \not= i\}$;
\item $\{(\text{check}_i,(1)_{j,d},\text{check}_i)\ |\ 1 \le j \le d\}$;
\item $\{(q_f,-x_f,\bot)\} \cup \{(p,0,\bot)\ |\ p \in \{\bot,\text{check},\text{check}_1,\dots,\text{check}_d\}\}$.
\end{itemize}

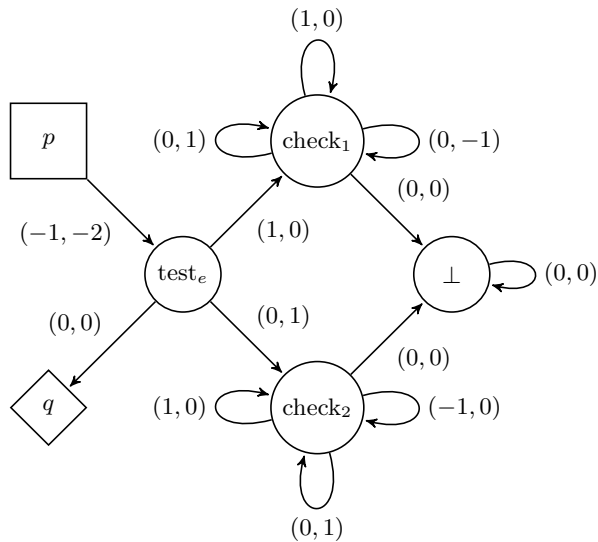
\begin{figure}
\begin{center}
\begin{tikzpicture}[->,>=stealth',shorten >=1pt,auto,node distance=2.5cm,
                    semithick]
  \tikzstyle{state}=[circle,minimum size=1cm,fill=white,draw=black,text=black,font=\small]
  \tikzstyle{oppstate}=[rectangle,minimum size=1cm,fill=white,draw=black,text=black,font=\small]
  \tikzstyle{anystate}=[diamond,minimum size=1cm,fill=white,draw=black,text=black,font=\small]

  \node[oppstate] (A)                    {$p$};
  \node[state]         (B) [below right of=A] {test$_e$};
  \node[state]         (C) [above right of=B] {check$_1$};
  \node[state]         (D) [below right of=B] {check$_2$};
  \node[state]         (E) [above right of=D] {$\bot$};
  \node[anystate]         (F) [below left of=B] {$q$};
  
  \path (A) edge [swap] node {$(-1,-2)$} (B)
        (B) edge [swap] node {$(1,0)$} (C)
            edge node {$(0,1)$} (D)
            edge [swap] node {$(0,0)$} (F)
        (C) edge [loop above] node {$(1,0)$} (C)
            edge [loop right] node {$(0,-1)$} (C)
            edge [loop left] node {$(0,1)$} (C)
            edge node {$(0,0)$} (E)
        (D) edge [loop below] node {$(0,1)$} (D)
            edge [loop left] node {$(1,0)$} (D)
            edge [loop right] node {$(-1,0)$} (D)
            edge [swap] node {$(0,0)$} (E)
        (E) edge [loop right] node {$(0,0)$} (E);

\end{tikzpicture}
\end{center}
\caption{Gadget to replace an edge $e = (p,(-1,-2),q)$ from an Opponent location
in the reduction from \pbvass{2}{}{(0,0)} to \pbcs{2}{}{(0,0)}.}
\label{vass_crg_opp}
\end{figure}

The objective of the counter reachability game is $(\bot,(0,\dots,0))$.
Hence, in the location check, Reacher has a winning strategy if, and only if, every counter is nonnegative,
and in the location check$_i$, Reacher has a winning strategy if, and only if, the $i^\text{th}$ counter,
which has been incremented when the play reached check$_i$, is nonpositive.
Consequently, as soon as a player makes a counter become negative, his adversary has a winning strategy by going to a check location.
If all counters remain positive, then Reacher has a winning move once the play visits the objective of the game on $(Q,E)$,
and only in this case.

The reduction is polynomial: we have $|Q'| \le d+2+|Q|+|E|$ and $|E'| \le (d+2)|E|+2d(d+1)+2$.
Moreover, the \shrtrg\ property is preserved when it holds for the reduced \vass,
provided that the objective in the \vass\ is a vector that contains only values in $\{-1,0,1\}$.\qed
\end{proof}

\subsection{Undecidability of counter reachability games on \vass}

The following proposition rephrases Proposition~$4$ from \cite{BJK10}.

\begin{theorem}[\cite{BJK10}]\label{eVASS_undec}
Let $(Q,E)$ be a \shrtrgd\ \vass\ of dimension two.
Consider a reachability game on $(Q,E)$ with $Q_Z \times ((\{0\} \times \N) \cup (\N \times \{0\}))$ as objective,
where $Q_Z \subseteq Q$. The problem of deciding the winner of this game is undecidable.
\end{theorem}

To apply Proposition~\ref{VASS_crg}, there must be only one configuration in the objective. 

\begin{proposition}\label{eVASS_VASS}
Let $(Q,E)$ be a \vass\ of dimension two. Consider a reachability game on $(Q,E)$
with $Q_Z \times ((\{0\} \times \N) \cup (\N \times \{0\}))$ as objective, where $Q_Z \subseteq Q$.
We can build a \vass\ $(Q',E')$ such that Reacher wins the reachability game on $(Q,E)$ if,
and only if, he wins the reachability game on $(Q',E')$ with objective $(\bot,(0,0))$, where $\bot \in Q' \setminus Q$.
\end{proposition}

\begin{proof}
We suppose that $Q_Z$ contains Reacher locations only.
This is without loss of generality as in the proof of Proposition~\ref{VASS_crg}. 
Let $Q' = Q \cup \{\emptyone,\emptytwo,\bot\}$, and let
\begin{align*}
E' = E & \cup \{(q,(0,0),\emptyone), (q,(0,0),\emptytwo)\ |\ q \in Q_Z\}\\
& \cup \{(\emptyone,(-1,0),\emptyone), (\emptytwo,(0,-1),\emptytwo)\}\\
& \cup \{(\emptyone,(0,0),\bot), (\emptytwo,(0,0),\bot), (\bot,(0,0),\bot)\}.
\end{align*}
Note that the \shrtrg\ property is preserved.
If Reacher has a winning strategy in the game on $(Q,E)$,
then he can follow the same strategy on $(Q',E')$ and reach a configuration where the location is in $Q_Z$ and one of the two counters is zero.
At this point, he can go to the location where he resets the second counter and, after that, go to $\bot$ and win.
Conversely, if Reacher has a winning strategy in the game on $(Q',E')$,
then he can enforce that the play visits $Q_Z \times ((\{0\} \times \N) \cup (\N \times \{0\}))$,
as this is the only possibility to reach a location $\emptyi_i$ with the $(3-i)^{\text{th}}$ counter at zero and,
after that, to reach the objective.
\qed
\end{proof}

\begin{theorem}\label{crg2_undec}
\pbcs{2}{1}{} is undecidable.
\end{theorem}

\begin{proof}
We make two successive reductions from the decision problem of Theorem~\ref{eVASS_undec} 
using Propositions~\ref{eVASS_VASS} and~\ref{VASS_crg}. 
\qed
\end{proof}

\section{Counter reachability games in dimension one}

In \cite{BJK10}, counter reachability games are played on \shrtrgd\ \vass,
where the winning condition in the one-dimensional case is to reach $Q_Z \times \{0\}$ for a given subset $Q_Z$ of $Q$.
It can be seen as the objective $(\bot,0)$, once we add a gadget that permits Reacher
to go from any location in $Q_Z$ to $\bot$ without any further modification of the counter value.
The decision problem is \psc\ in general and it is in \ptm\ when $Q_Z = Q$.

In this section, we establish mutual reductions between the decision problem for counter reachability games
under the three semantics in dimension one. The~complexity classes follow from the reductions.

\subsection{Relative integers semantics}

We recall that Proposition~\ref{VASS_crg} implies that there is a polynomial-time reduction from \pbvass{1}{1}{0} 
to \pbcs{1}{1}{0}, hence \pbcs{1}{1}{0} is \psh.

The main idea of the construction in this section
is to simulate, with nonnegative integers only, a counter value in $\Z$.
For this purpose, we use two copies of the set of locations
and explain how to move from one copy to another.

\begin{theorem}\label{crg1_vass}
\pbcs{1}{1}{0} is \psc.
\end{theorem}

\begin{proof}
We reduce \pbcs{1}{1}{0} to \pbvass{1}{1}{0} in polynomial time.
Consider a~reachability game on a \shrtrgd\ counter system $(Q,E)$,
where the objective is $(q_f,0)$, with $q_f \in Q_1$.
Note that when the objective counter value is not $0$, we can always shift initial and objective value in a general counter system.

Let $Q_+ = \{q_+\ |\ q \in Q\}$ and $Q_- = \{q_-\ |\ q \in Q\}$ be two copies of $Q$,
and let $Q_E$ be the set $\{q_e\ |\ \exists p, q \in E, v \in \{\pm 1\}, e = (p,v,q) \in E\}$.
We build the \shrtrgd\ \vass\ $(Q',E')$, where $Q' = Q_+ \cup Q_- \cup Q_E \cup \{\text{no},\bot\}$ is partitioned into
$Q'_1 = \{q_+,q_-\ |\ q \in Q_1\} \cup \{q_e \in Q_E\ |\ e \in Q_2 \times \{0,\pm 1\} \times Q\} \cup \{\text{no},\bot\}$ and $Q'_2$.
The set of edges $E'$ contains two copies of $E$, i.e., edges $(p_+,v,q_+)$ and $(p_-,-v,q_-)$ for each edge $(p,v,q) \in E$.
The other edges of $E'$ are used to move between $Q_+$ and $Q_-$ via the new locations of $Q_E$,
as depicted in Figures~\ref{crg_vass_reach} and~\ref{crg_vass_opp}. 

More precisely, $E'$ is the union of the following sets of edges:
\begin{itemize}
\item $\{(p_+,v,q_+),(p_-,-v,q_-)\ |\ (p,v,q) \in E\}$;
\item $\{(p_-,0,q_e), (q_e,0,\bot), (q_e,1,q_+)\ |\ e = (p,1,q) \in E, p \in Q'_1\}$;
\item $\{(p_+,0,q_e), (q_e,0,\bot), (q_e,1,q_-)\ |\ e = (p,-1,q) \in E, p \in Q'_1\}$;
\item $\{(p_-,0,q_e), (q_e,-1,\text{no}), (q_e,1,q_+)\ |\ e = (p,1,q) \in E, p \in Q'_2\}$;
\item $\{(p_+,0,q_e), (q_e,-1,\text{no}), (q_e,1,q_-)\ |\ e = (p,-1,q) \in E, p \in Q'_2\}$;
\item $\{(\text{no},-1,\text{no}), (\text{no},0,\bot), (q_{f,+},0,\bot), (q_{f,-},0,\bot), (\bot,0,\bot)\}$.
\end{itemize}

The \vass\ $(Q',E')$ is designed such that a play in it corresponds to a play in the counter system $(Q,E)$.
Hence, a configuration $(q,x) \in Q \times -\N$ in $(Q,E)$
is associated to the configuration $(q_-,-x) \in Q_- \times \N$ in $(Q',E')$.
That is why the labels of the edges between locations in $Q_-$ are the opposite of the labels of the edges in $Q$.

The objective of the game on~$(Q',E')$ is~$(\bot,0)$.
In fact, Reacher loses whenever a play reaches $\bot$ with another counter value.
Furthermore, if a player makes a move to a location $q_e$ in $Q_E$ and the counter value is not~$0$,
then his adversary, who owns $q_e$, has a winning move.
\qed

\begin{figure}
\begin{center}
\begin{tikzpicture}[->,>=stealth',shorten >=1pt,auto,node distance=2cm,
                    semithick]
  \tikzstyle{state}=[circle,minimum size=8mm,fill=white,draw=black,text=black,font=\small]
  \tikzstyle{oppstate}=[rectangle,minimum size=8mm,fill=white,draw=black,text=black,font=\small]
  \tikzstyle{anystate}=[diamond,minimum size=8mm,fill=white,draw=black,text=black,font=\small]

  \node[state] (A)                    {$p_-$};
  \node[state] (B) [below of =A]      {$p_+$};
  \node[anystate] (C) [right of=A]    {$q_-$};
  \node[oppstate] (D) [left of=B] {$q_e$};
  \node[anystate] (E) [right of=B] {$q_+$};
  \node[state] (F) [above of=D] {$\bot$};

  \path (A) edge node {$1$} (C)
        (B) edge node {$-1$} (E)
            edge node {$0$} (D)
        (D) edge node {$0$} (F)
            edge node {$1$} (C)
        (F) edge [loop left] node {$0$} (F);

\end{tikzpicture}
\end{center}
\caption{Gadget to replace an edge $e = (p,-1,q)$ from a Reacher location
in the reduction from \pbcs{1}{1}{0} to \pbvass{1}{1}{0}.}
\label{crg_vass_reach}
\end{figure}

\begin{figure}
\begin{center}
\begin{tikzpicture}[->,>=stealth',shorten >=1pt,auto,node distance=2cm,
                    semithick]
  \tikzstyle{state}=[circle,minimum size=8mm,fill=white,draw=black,text=black,font=\small]
  \tikzstyle{oppstate}=[rectangle,minimum size=8mm,fill=white,draw=black,text=black,font=\small]
  \tikzstyle{anystate}=[diamond,minimum size=8mm,fill=white,draw=black,text=black,font=\small]

  \node[oppstate] (A)                    {$p_+$};
  \node[oppstate] (B) [below of =A]      {$p_-$};
  \node[anystate] (C) [right of=A]    {$q_+$};
  \node[state] (D) [left of=B] {$q_e$};
  \node[anystate] (E) [right of=B] {$q_-$};
  \node[state] (F) [above of=D] {no};
  \node[state] (G) [left of=F] {$\bot$};

  \path (A) edge node {$1$} (C)
        (B) edge node {$-1$} (E)
            edge node {$0$} (D)
        (D) edge node {$-1$} (F)
            edge node {$1$} (C)
        (F) edge [loop right] node {$-1$} (F)
            edge node {$0$} (G)
        (G) edge [loop below] node {$0$} (G);

\end{tikzpicture}
\end{center}
\caption{Gadget to replace an edge $e = (p,1,q)$ from an Opponent location
in the reduction from \pbcs{1}{1}{0} to \pbvass{1}{1}{0}.}
\label{crg_vass_opp}
\end{figure}
\end{proof}

A consequence of Theorem~\ref{crg1_vass} is that \pbcs{1}{}{} is in \xps:
It suffices to split every edge with another label than $-1$, $0$ or $1$.
However, we do not know yet whether \xps\ is an optimal upper bound, but we have the following lower bound.

\begin{theorem}[\cite{JLS07,AR13}]\label{crg1_expthard}
\pbcs{1}{}{} is \xpth.
\end{theorem}

This lower bound is inherited from countdown games \cite{JLS07} and robot games \cite{AR13},
which we can express as counter reachability games.

\subsection{\Wkvass\ semantics}

When we simulate a game on a \wkvass, we need, like for \vass, to handle the behaviour around the value $0$.
The idea is the following: For every edge labelled by $-1$ in a \shrtrgd\ \wkvass,
there are two choices for Opponent in the \vass: decrement the counter or leave it unchanged,
depending on whether it is positive or zero.
The winning condition is designed so that Reacher has a checking move that makes him win
whenever Opponent chooses the wrong move, e.g., he leaves the counter unchanged whereas he should decrement it.
Moreover, Opponent wins if Reacher abuses his checking move.

\begin{theorem}\label{wkvass_VASS}
\pbwkvass{1}{1}{1} is \psc.
\end{theorem}

\begin{proof}[\psp-hardness]
We reduce \pbwkvass{1}{1}{1} to \pbvass{1}{1}{0} in polynomial time.
Consider a~reachability game on a \shrtrgd\ \wkvass\ $(Q,E)$, where the objective is $(q_f,1)$,
with the assumption that $q_f \in Q_1$.
Let $Q_E$ be the set $\{q_e,q^{>0}_e,q^{=0}_e\ |\ e \in E \cap (Q \times \{-1\} \times Q)\}$.
We build the \shrtrgd\ \vass\ $(Q',E')$, where $Q' = Q \cup Q_E \cup \{\text{no},\bot\}$ is partitioned into
$Q'_1 = Q_1 \cup \{q^{>0}_e,q^{=0}_e\ |\ e \in E\} \cup \{\text{no},\bot\}$ and $Q'_2$. The set of edges~is
\begin{align*}
E' = & \{(p,v,q)\ |\ (p,v,q) \in E, v \in \{0,1\}\}\\
& \cup\ \{(p,0,q_e), (q_e,0,q^{>0}_e), (q_e,0,q^{=0}_e), (q^{=0}_e,0,q), (q^{>0}_e,-1,q),\\
& \hspace*{5mm} (q^{>0}_e,0,\bot), (q^{=0}_e,-1,\bot)\ |\ e = (p,-1,q) \in E\}\\
& \cup\ \{(\text{no},-1,\text{no}), (\text{no},0,\bot), (q_f,-1,\bot), (\bot,0,\bot)\}.
\end{align*}

Intuitively, every time a play visits an edge with a decrement in $(Q',E')$,
Opponent has to guess whether the counter value is zero or positive, and move accordingly to an intermediate location,
where Reacher can move to the actual target of the edge in $(Q,E)$ or to a checking module where the game ends.

The objective of the game on~$(Q',E')$ is~$(\bot,0)$.
As we can see in Figure~\ref{wkvass_vass_fig}, Reacher has a winning strategy in every location $q^{=0}_e$ 
when the counter value is positive, and in every location $q^{>0}_e$ when the counter value is zero.
\qed

\begin{figure}
\begin{center}
\begin{tikzpicture}[->,>=stealth',shorten >=1pt,auto,node distance=2cm,
                    semithick]
  \tikzstyle{state}=[circle,minimum size=8mm,fill=white,draw=black,text=black,font=\small]
  \tikzstyle{oppstate}=[rectangle,minimum size=8mm,fill=white,draw=black,text=black,font=\small]
  \tikzstyle{anystate}=[diamond,minimum size=8mm,fill=white,draw=black,text=black,font=\small]

  \node[anystate] (A)                    {$p$};
  \node[oppstate] (B) [right of =A]      {$q_e$};
  \node[state] (C) [above right of=B]    {$q^{=0}_e$};
  \node[state] (D) [below right of=B] {$q^{>0}_e$};
  \node[anystate] (E) [right of=B] {$q$};
  \node[state] (F) [right of=E] {no};
  \node[state] (G) [right of=F] {$\bot$};

  \path (A) edge node {$0$} (B)
        (B) edge node {$0$} (C)
            edge [swap] node {$0$} (D)
        (C) edge [swap] node {$0$} (E)
            edge node {$-1$} (F)
        (D) edge node {$-1$} (E)
            edge [swap] node {$0$} (G)
        (F) edge [loop above] node {$-1$} (F)
            edge node {$0$} (G)
        (G) edge [loop above] node {$0$} (G);

\end{tikzpicture}
\end{center}
\caption{Gadget to replace an edge $e = (p,-1,q)$
in the reduction from \pbwkvass{1}{1}{1} to \pbvass{1}{1}{0}.}
\label{wkvass_vass_fig}
\end{figure}
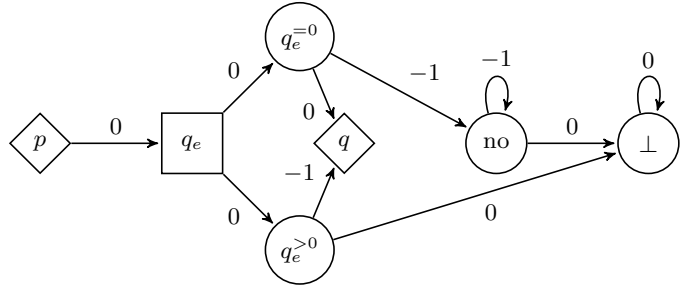
\end{proof}

In the construction for the reverse reduction, when a player chooses any edge with a negative label
and the counter value is less than the value that should be subtracted,
then the adversary of this player has a winning move.
Whereas this is no problem in a \wkvass, such an edge would be forbidden in a \vass.

\begin{proof}[\psp-membership]
We show a polynomial-time reduction, that preserves the \shrtrg\ property, from \pbvass{1}{}{0} to \pbwkvass{1}{}{1}.
Consider a reachability game on a \vass\ $(Q,E)$, where the objective is $(q_f,0)$,
with $q_f \in Q_1$. Let $Q_E$ be the set $\{q_e\ |\ e \in E \cap (Q \times (\Z \setminus \N) \times Q)\}$.
We build the \wkvass\ $(Q',E')$, where $Q' = Q \cup Q_E \cup \{\text{no}_\text{R},\text{no}_\text{O},\bot\}$,
$Q'_1 = Q_1 \cup \{q_e \in Q_E\ |\ e \in Q_2 \times \Z \times Q\} \cup \{\text{no}_\text{R},\text{no}_\text{O},\bot\}$, $Q'_2 = Q' \setminus Q'_1$, and
$E'$ is obtained from $E$ by splitting every edge $(p,v,q)$ such that $v \in -\N$ into two edges $(p,0,q_e)$ and $(q_e,v,q)$
and by adding an edge from every location $q_e$ to the ``no''-location that corresponds to the owner of $p$,
as well as additional edges between no$_O$, no$_R$ and $\bot$, as depicted in the Figures~\ref{vass_wkvass_reach} and~\ref{vass_wkvass_opp}. 

More precisely, $E'$ is the union of the sets of edges:
\begin{itemize}
\item $\{(p,v,q)\ |\ (p,v,q) \in E, x \in \N\}$;
\item $\{(p,0,q_e), (q_e,v,q)\ |\ e = (p,v,q) \in E, x < 0\}$;
\item $\{(q_e,x+1,\text{no}_\text{R})\ |\ e = (p,v,q) \in E, v < 0, p \in Q_1\}$;
\item $\{(q_e,x+1,\text{no}_\text{O})\ |\ e = (p,v,q) \in E, v < 0, p \in Q_2\}$;
\item extra edges $\{(\text{no}_\text{R},-1,\text{no}_\text{R}), (\text{no}_\text{R},0,\bot), (\text{no}_\text{O},1,\bot), (q_f,1,\bot), (\bot,0,\bot)\}$.
\end{itemize}

\begin{figure}
\begin{center}
\begin{tikzpicture}[->,>=stealth',shorten >=1pt,auto,node distance=2cm,
                    semithick]
  \tikzstyle{state}=[circle,minimum size=8mm,fill=white,draw=black,text=black]
  \tikzstyle{oppstate}=[rectangle,minimum size=8mm,fill=white,draw=black,text=black]
  \tikzstyle{anystate}=[diamond,minimum size=8mm,fill=white,draw=black,text=black]

  \node[state] (A)                    {$p$};
  \node[oppstate]         (B) [right of=A] {$q_e$};
  \node[state]         (C) [right of=B] {no$_\text{R}$};
  \node[state]         (D) [right of=C] {$\bot$};
  \node[anystate]         (E) [below of=B] {$q$};

  \path (A) edge node {$0$} (B)
        (B) edge node {$-4$} (C)
            edge node {$-5$} (E)
        (C) edge [loop below] node {$-1$} (C)
            edge node {$0$} (D)
        (D) edge [loop below] node {$0$} (D);

\end{tikzpicture}
\end{center}
\caption{Gadget to replace an edge $e = (p,-5,q)$ from a Reacher location
in the reduction from \pbvass{1}{1}{0} to \pbwkvass{1}{1}{1}.}
\label{vass_wkvass_reach}
\end{figure}

\begin{figure}
\begin{center}
\begin{tikzpicture}[->,>=stealth',shorten >=1pt,auto,node distance=2cm,
                    semithick]
  \tikzstyle{state}=[circle,minimum size=8mm,fill=white,draw=black,text=black]
  \tikzstyle{oppstate}=[rectangle,minimum size=8mm,fill=white,draw=black,text=black]
  \tikzstyle{anystate}=[diamond,minimum size=8mm,fill=white,draw=black,text=black]

  \node[oppstate] (A)                    {$p$};
  \node[state]         (B) [right of=A] {$q_e$};
  \node[state]         (C) [right of=B] {no$_\text{O}$};
  \node[state]         (D) [right of=C] {$\bot$};
  \node[anystate]         (E) [below of=B] {$q$};
  
  \path (A) edge node {$0$} (B)
        (B) edge node {$-4$} (C)
            edge node {$-5$} (E)
        (C) edge node {$1$} (D)
        (D) edge [loop below] node {$0$} (D);

\end{tikzpicture}
\end{center}
\caption{Gadget to replace an edge $e = (p,-5,q)$ from an Opponent location
in the reduction from \pbvass{1}{1}{0} to \pbwkvass{1}{1}{1}.}
\label{vass_wkvass_opp}
\end{figure}

The \wkvass\ $(Q',E')$ is designed such that a play in it corresponds to a play in the \vass\ $(Q,E)$.
Let us consider a location $q_e \in Q_E$, for an edge $(p,v,q)$ in $E$.
Note that $v < 0$ and that the owner of $q_e$ is not the owner of $p$.
In~the play on the \vass, the edge $(p,v,q)$ can only be taken if the counter value is at least $-v$.
If a player goes to $q_e$, i.e., simulates the choice of the edge $(p,v,q)$,
his adversary should win whenever the counter value is less than $-v$,
by going to a ``no''-location, as we can see in the Figures~\ref{vass_wkvass_reach} and~\ref{vass_wkvass_opp}. 
\qed
\end{proof}

\subsection{The case of zero-reachability on \wkvass}

For \wkvass, we prove that the set of winning configurations is downward closed when the reachability objective is $(q_f,0)$ for a given $q_f$.
Hence, to decide whether Reacher has a winning strategy, we compute for all locations the maximal initial value
for which the pair (location, value) is winning and we look at the initial configuration.

\begin{lemma}\label{wkvass_dwnwcl}
Let $(Q,E)$ be a \wkvass. Consider a reachability game on $(Q,E)$, where the objective is $(q_f,0)$, where $q_f \in Q$.
If the initial configuration $(q_0,x)$ is winning, then every configuration $(q_0,x')$ for $x' < x$ is winning.
\end{lemma}

\begin{proof}
Let $(q_0,x)$ be a winning configuration, and let $s$ be a winning strategy for Reacher from $(q_0,x)$.
Consider any strategy $s'$ for Opponent.
The outcome of the strategies $s$ and $s'$ from $(q_0,x)$ is a play $\pi$ that Reacher wins, i.e., the play $\pi$ eventually visits $(q_f,0)$.
Now, let us look at the outcome of the strategies $s$ and $s'$ from $(q_0,x')$ for $x' < x$.
It is a play $\pi'$ that visits the same locations as $\pi$, and no edge is disabled because of the semantics of a \wkvass.
Moreover, the counter value in $\pi'$ is after each move less than or equal to the counter value in the corresponding move of $\pi$.
In particular, $\pi'$ eventually visits $q_f$ with counter value $0$, hence Reacher wins.
\qed
\end{proof}

Algorithm \ref{algo} determines the winner of a reachability game on a \wkvass\ when the objective counter value is~$0$. 
Its time complexity is exponential in the initial counter value.
Accordingly, we call it only with~$0$ as initial counter value in the proof of Theorem~\ref{wkvass_0_P}. 

\begin{algorithm}
\KwData{A \wkvass\ $(Q,E)$, a location $q_f$, and a configuration $(q_0,x_0)$}
\KwResult{Does Reacher have a winning strategy to reach $(q_f,0)$ from $(q_0,x_0)$?}
\Begin{
Create a table $M_q$ with $q \in Q$ as indices initialized to $-\infty$\;
$M_{q_f}\gets 0$\;
\Repeat{a fixpoint is reached or $M_{q_0} \ge x_0$}{
  \ForEach{$e = (q,v,q')) \in E$}{$M_q \gets \max(M_q,M_q'-v)$}
}
\lIf{$M_{q_0} \ge x_0$}{\Return{true}}\;
\lElse{\Return{false}}\;
}
\label{algo}
\caption{Solves \pbwkvass{1}{}{0}.}
\end{algorithm}

\begin{theorem}\label{wkvass_0_P}
\pbwkvass{1}{1}{0} is in \ptm.
\end{theorem}

\begin{proof}
According to Lemma~\ref{wkvass_dwnwcl}, we just need to compute for every location $q \in Q$ the maximal value $x_m$ such that $(q,x_m)$ is winning. 
We even do more: First, we compute the set $Q_Z$ of locations from which Reacher has a winning strategy with initial counter value $0$.
For this purpose, we use the previous algorithm, and here the time complexity is polynomial.
Second, we build the \vass\ $(Q',E')$, where $Q' = Q_Z \cup \{\bot\}$
and $E'$ is the union of $E \cap (Q_Z \times \Z \times Q_Z)$ and of $\{(q,1,\bot)\ |\ (q,v,q') \in E, q \in Q_Z, q' \not\in Q_Z\} \cup \{(\bot,0,\bot)\}$.
In $(Q',E')$, the value $0$ can only be reached in a location that belongs to $Q_Z$.
Consider the reachability game on $(Q',E')$, where the objective is $Q \times 0$,
like defined in~\cite{BJK10}; deciding the winner in this game is in~\ptm.
Moreover, Reacher has a winning strategy if, and only if,
he has a strategy in $Q$ to reach $(q,0)$ for any $q \in Q_Z$, hence to reach $(q_f,0)$.
Indeed, if a play visits a location outside of $Q_Z$, then Opponent has a winning strategy.
We conclude that deciding the winner of the reachability game is in \ptm\ too.\qed
\end{proof}

Note that we need the \shrtrg\ property for our \wkvass, else the algorithm could still require exponential time.
For example, consider that there is an edge from $q_0$ to $q_f$ with label $2^n$ and a self-loop on $q_f$ with label $-1$.
The algorithm would need $2^n+1$ iterations to conclude that $(q_0,0)$ is a winning configuration,
whereas the size of the \wkvass\ is linear in $n$ because of the binary encoding.

\section{Conclusion}

In this paper, we studied three simple semantics for games on counter systems,
and compared the complexity of reachability problems.
In dimension~two, every problem that we considered is undecidable.
In dimension~one, the decision problems associated to the counter value~$0$ are
in \ptm\ for the case of the \wkvass\ semantics and \psc\ for the two other semantics,
when the counter system is \shrtrgd.
Without this property, which guarantees that the set of all visited counter values is an interval,
the complexity is not settled yet, to the best of our knowledge, and lies between \xpt\ and \xps.

\vspace{5mm}
\textbf{Acknowledgement.} The author would like to thank Dietmar Berwanger and Laurent Doyen
for proposing the topic and for helping to organize the paper,
and Marie van den Bogaard for patient reading and checking of the proofs.

\bibliographystyle{unsrt}
\bibliography{CRG}
\end{document}